\documentclass[a4paper]{llncs}
\usepackage[utf8]{inputenc}
\usepackage{amsmath,amssymb}
\usepackage{tikz}
\usepackage{hyperref}
\usepackage{wrapfig}

\newcommand{\defo}[1]{\emph{#1}}

\newcommand{\abovebelow}[2]{(^{#1}_{#2})}

\newcommand{\rank}{r}
\newcommand{\inv}{s}
\newcommand{\ranki}{\mathring{r}}
\newcommand{\invi}{\mathring{s}}
\newcommand{\tr}[1]{{#1}\!\strut^\intercal\!}
\newcommand{\BMS}{\text{\tiny\textsc{BMS}}}

\definecolor{rkBlue}{RGB}{50,20,230}
\definecolor{siGreen}{RGB}{0,103,0}

\title{Linear Ranking for Linear Lasso Programs\thanks{
The final publication is available at \href{http://link.springer.com/chapter/10.1007\%2F978-3-319-02444-8_26}{link.springer.com}.
} \thanks{
  This work is supported by the
  German Research Council (DFG) as part of the Transregional Collaborative
  Research Center ``Automatic Verification and Analysis of Complex Systems''
  (SFB/TR14 AVACS)
}
}

\author{Matthias Heizmann \and Jochen Hoenicke\and Jan Leike \and Andreas Podelski}

\institute{University of Freiburg, Germany}

\begin{document}

\maketitle              

\sloppy

\begin{abstract}
The general setting of this work is the constraint-based synthesis of termination arguments.
We consider a restricted class of programs called lasso programs. The termination argument for a lasso program is a pair of a ranking function and an invariant.
We present the---to the best of our knowledge---first method to synthesize termination arguments for lasso programs that uses  linear arithmetic.
We prove a completeness theorem.
The completeness theorem establishes that, even though we use only linear (as opposed to non-linear) constraint solving, we are able to compute termination arguments in several interesting cases.
The key to our method lies in a constraint transformation that replaces a disjunction by a sum.
\end{abstract}

\section{Introduction}
Termination is arguably the single most interesting correctness property of a program.
Research on proving termination can be divided according to three (interrelated) topics, namely:
practical tools~\cite{DBLP:journals/jar/AlbertAGP11,cav/BrockschmidtMOG12,cav/CookPR06,sas/HarrisLNR10,atva/KroeningSTTW08,cav/KroeningSTW10,lics/PodelskiR04,hybrid/PodelskiW07},
\mbox{decidability} questions~\cite{vmcai/Ben-AmramGM12,cav/Braverman06,cav/Tiwari04}, 
and
constraint-based synthesis of
termination \mbox{arguments}
~\cite{DBLP:journals/iandc/BagnaraMPZ12,DBLP:journals/corr/abs-1208-4041,cav/BradleyMS05,icalp/BradleyMS05,concur/BradleyMS05,tacas/ColonS01,journals/fmsd/CookKRW13,vmcai/Cousot05,vmcai/PodelskiR04,cav/Rybalchenko10}.  The work in this paper falls under the research on the third topic.
The general goal of this research is to investigate how one can derive a constraint from the program text and compute a termination argument (of a restricted form) through the solution of the constraint, i.e., via constraint solving.

In this paper, we present a method for the synthesis of termination arguments for a specific class of programs that we call \emph{lasso programs}. 
As the name indicates, the control flow graph of a lasso program is of a restricted shape: a \emph{stem} followed by a \emph{loop}.

Lasso programs do not appear as stand-alone programs.
Lasso programs appear in practice whenever one needs a finite representation of an infinite path in a control flow graph, for example in (potentially spurious) counterexamples in a termination analysis\cite{cav/CookPR06,sas/HarrisLNR10,atva/KroeningSTTW08,cav/KroeningSTW10},  non-termination analysis\cite{conf/popl/GuptaHMRX08}, stability analysis\cite{vmcai/CookFKP11,hybrid/PodelskiW07}, or cost analysis\cite{DBLP:journals/jar/AlbertAGP11,conf/pldi/GulwaniZ10}.

Importantly, the termination argument for a lasso program is a pair of a ranking function and an invariant (the rank must decrease only for states that satisfy the invariant). \autoref{fig-bangalore} shows an example of a lasso program.

The class of lasso programs lies between two classes of programs for which constraint-based methods have been studied extensively. 
For the first, more specialized class, methods can be based on linear arithmetic constraint solving~\cite{DBLP:journals/iandc/BagnaraMPZ12,DBLP:journals/corr/abs-1208-4041,tacas/ColonS01,journals/fmsd/CookKRW13,vmcai/PodelskiR04}. 
For the second, more general class, all known methods are based on non-linear arithmetic constraint solving~\cite{cav/BradleyMS05,concur/BradleyMS05}. 
The contribution of our method can be phrased, alternatively, as the generalization of the applicability of the `linear methods', or as the
optimization of the `non-linear method' to a `linear method' for a subproblem.
The step from `non-linear' to `linear' is interesting for principled reasons (non-linear arithmetic constraint solving is undecidable in the case of integers).
As we will show the step is also practically interesting.

The reader may wonder how practical tools presently handle the situation where one needs to compute  termination arguments for lasso programs. 
One possibility is to resort to heuristics.
  For example, instead of computing a termination argument for the lasso program in Figure 1, one would compute the ranking function $f(x)=x$ for the program \verb|while(x>=0){x:=x-23;}|.

The key to our method is a constraint transformation that replaces a disjunction by a sum. We apply the `or-to-plus' transformation in the context of Farkas' Lemma.  
Following~\cite{DBLP:journals/iandc/BagnaraMPZ12,cav/BradleyMS05,tacas/ColonS01,journals/fmsd/CookKRW13,vmcai/PodelskiR04}, we apply Farkas' Lemma in order to eliminate the universal quantifiers in the arithmetic constraint whose solution is the termination argument.
If we apply Farkas' Lemma to the constraint \emph{after} the `or-to-plus' transformation, we obtain a \emph{linear} arithmetic constraint.

The effect of the `or-to-plus' transformation to the constraint is a restriction of its solution space.  The restriction seems strong; i.e., in some cases, the solution space becomes empty.
We can characterize those cases. In other words, we can characterize when the `or-to-plus' transformation leads to the loss of an termination argument, and when it does not.
The characterization is formulated as a completeness theorem for which we will present the proof.
This characterization allows us to establish that, even though we use only linear (as opposed to non-linear) constraint solving, we are able to compute termination arguments in several interesting cases.
A possible explanation for this (perhaps initially surprising) fact is that, for synthesis, we are interested in the mere existence of a solution, and the loss of \emph{many} solutions does not necessarily mean the loss of \emph{all} solutions of the constraint.

\begin{figure}[t]
\begin{center}
\begin{minipage}{4cm}

\begin{verbatim}
1: y := 23;
2: while( x >= 0 ) {
3:     x := x - y;
4:     y := y + 1;
5: }
\end{verbatim}

\end{minipage}
\begin{minipage}{6cm}
\begin{tikzpicture}[auto,
trans/.style={->,>=stealth,thick}]
 \node (0) at (0,0) [circle,draw,inner sep=1.5] {\scriptsize \texttt{1}};
 \node (1) at (2,0) [circle,draw,inner sep=1.7] {\scriptsize \texttt{2}};
\draw [trans] (0) to node[] {$\begin{array}{c}\tau_\mathsf{stem}:\\y' = 23 \end{array}$} (1);
\draw [trans,rounded corners=10mm, pos=0.1] (1) -- (3.3,1) -- node {$\tau_\mathsf{loop}:\;\;\begin{array}{rl}& x \geq 0\\ \land & x'=x-y\\ \land & y'=y+1\end{array}$} (3,-1) --  (1);

\end{tikzpicture}
\end{minipage}
\end{center}
\caption{Example of a lasso program and its formal representation $P_\mathsf{yPositive}=(\tau_\mathsf{stem},\tau_\mathsf{loop})$.
The ranking function defined by ${\color{rkBlue}f(x,y)=x}$ decreases in transitions from states that satisfy the invariant ${\color{siGreen}y\geq 1}$ (the ranking function does not decrease when $y\leq 0$).
}
\label{fig-bangalore}
\end{figure}
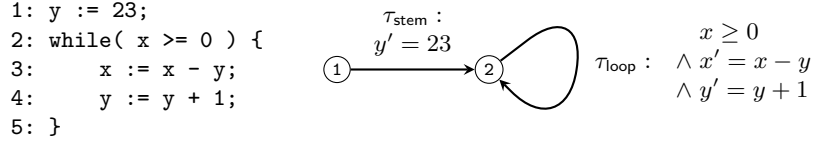

We have implemented our method and we have used our implementation to illustrate the applicability and the efficiency of our method. Our implementation is available through a web interface, together with a number of example programs (including the ones used in this paper).\footnote{
\url{http://ultimate.informatik.uni-freiburg.de/LassoRanker}
}

\section{Preliminaries: Linear Arithmetic}
We use $\vec x$ to denote the vector with entries $x_1,\dots,x_n$, and $\tr{\vec x}$ to denote the transposed vector of $\vec x$.
As usual, the expression $A\cdot\vec x \leq \vec b$ denotes the conjunction of linear constraints $\bigwedge\limits_{j=0}^m ( \sum\limits_{i=0}^n a_{ij}\cdot x_i)\leq b_j$.

We call a relation $\tau(\vec x,\vec x')$ a \defo{linear relation} if $\tau$ is defined by a conjunction of linear constraints over the variables $\vec x$ and $\vec x'$, i.e., if there is a matrix $A$ with $m$ rows and $2n$ columns and a vector $\vec b$ of size $m$ such that the following equation holds.
$$\tau(\vec x,\vec x')=\{(\vec x,\vec x')\mid A\cdot\abovebelow{\vec x}{\vec x'}\leq \vec b\}$$

We call a function $f(\vec x)$ an \defo{(affine) linear function}, if $f(\vec x)$ is defined by an affine linear term, i.e., there is a vector $\tr{\vec r}$ and a number $r_0$ such that the following equation holds.
$$f(\vec x) = \tr{\vec r}\cdot \vec x + r_0.$$

We call a predicate $I(\vec x)$ a \defo{linear predicate}, if $I(\vec x)$ is defined by a linear inequality, i.e., there is a vector $\tr{\vec s}$ and a number $s_0$ such that following equivalence holds.
$$I(\vec x)=\{\vec x\mid \tr{\vec s}\cdot \vec x+s_0\geq 0\}.$$

\paragraph{Farkas' Lemma.} We use the affine version of Farkas' Lemma~\cite{Schrijver:1986:TLI:17634} which is also used in ~\cite{DBLP:journals/iandc/BagnaraMPZ12,cav/BradleyMS05,journals/fmsd/CookKRW13,cav/Rybalchenko10,vmcai/PodelskiR04} and states the following.
Given
\begin{itemize}
 \item a satisfiable conjunction of linear constraints $A\cdot \vec x\leq \vec b$
 \item and a linear constraint $\tr{\vec c}\cdot \vec x \leq \delta$,
\end{itemize}
the following equivalence holds.
\begin{center}
$\forall \vec x\;\;  (A\cdot \vec x\leq \vec b \rightarrow \tr{\vec c}\cdot \vec x \leq \delta)$ \ \ \ iff \ \ \  $\exists \vec\lambda\;\;(\vec\lambda\geq 0 \land \tr{\vec\lambda}\cdot A = \tr{\vec c} \land \tr{\vec\lambda}\cdot \vec b\leq \delta)$
\end{center}

\section{Lasso Program}
To abstract away from program syntax, we define a lasso program directly by the two relations that generate its execution sequences.

\begin{definition}[Lasso Program] Given a set of states $\Sigma$, a \defo{lasso program} 
$$P=(\tau_\mathsf{stem},\tau_\mathsf{loop})$$
is given by the two relations $\tau_\mathsf{stem}\subseteq \Sigma\times\Sigma$ and $\tau_\mathsf{loop}\subseteq \Sigma\times\Sigma$.
We call $\tau_\mathsf{stem}$ the \defo{stem} of $P$ and $\tau_\mathsf{loop}$ the \defo{loop} of $P$.

An \defo{execution of the lasso program} $P$ is a possibly infinite sequence of states $\sigma_0, \sigma_1, \ldots$ such that
\begin{itemize}
 \item the pair of the first two states is an element of the stem, i.e.,
 $$ (\sigma_0, \sigma_1)\in \tau_{stem}$$
 \item and each other consecutive pair of states is an element of the loop, i.e.,
 $$ (\sigma_i, \sigma_{i+1})\in \tau_{loop} \qquad \text{ for } i=1,2,\dots$$
\end{itemize}
We call the lasso program $P$ \defo{terminating} if $P$ has no infinite execution.

\end{definition}
We use constraints over primed and unprimed variables to denote a transition relation (see~\autoref{fig-bangalore}).

In order to avoid cumbersome technicalities, we consider only lasso programs that have an execution that contains at least three states. This means we consider only programs where the relational composition of $\tau_\mathsf{stem}$ and $\tau_\mathsf{loop}$ is non-empty, i.e.,
$$\tau_\mathsf{stem} \circ \tau_\mathsf{loop}\neq\emptyset.$$

Since Turing, a termination argument is based on an ordering which does not allow infinite decreasing chains (such as ordering on the natural numbers).  Here, we use the ordering over the set of positive reals which is defined by some value $\delta>0$, namely

\begin{center}
\hfil\hfil\hfil $a\prec_\delta b$ \ \ \ \  iff \ \ \ \  $a\geq 0 \;\land\; a-b\geq\delta$ \hfil\hfil\hfil  $a,b\in\mathbb{R}.$
\end{center}

\paragraph{Ranking Function.}
We call a function $f$ from the states of the lasso program $P$ into
the reals $\mathbb{R}$ a \defo{ranking function} for $P$ if there is a
positive number $\delta>0$ such that for each consecutive pair of
states $(\vec x_i, \vec x_{i+1})$ of a loop transition ($i \geq 1$)
in every execution of $P$
\begin{itemize}
\item the value of $f$ is decreasing by at least $\delta$, i.e.,
\begin{align*}
f(\vec x_i) - f(\vec x_{i+1})\geq\delta,
\end{align*}
\item and the value of $f$ is non-negative, i.e.,
\begin{align*}
f(\vec x_i)\geq 0.
\end{align*}
\end{itemize}
If there is a ranking function for the lasso program $P$, then $P$ is terminating.

\paragraph{Inductive Invariant.}
We call a state predicate ${\color{siGreen} I(\vec x)}$ an \defo{inductive invariant} of the lasso program $P$ if

\begin{itemize}
\item the predicate holds after executing the stem, i.e.,
\begin{displaymath}
\forall \vec x\;\forall \vec x'\quad \tau_\mathsf{stem}(\vec x,\vec x') \rightarrow {\color{siGreen} I(\vec x')}, \tag{$\varphi_\mathsf{invStem}$}\label{inv-stem1}
\end{displaymath}

\item and if the predicate holds before executing the loop, then the predicate holds afterwards, i.e.,
\[
\forall \vec x\;\forall \vec x'\quad{\color{siGreen} I(\vec x)} \;\land\;\tau_\mathsf{loop}(\vec x,\vec x') \rightarrow {\color{siGreen} I(\vec x')}.\tag{$\varphi_\mathsf{invLoop}$}\label{inv-loop1}
\]
\end{itemize}

\paragraph{Ranking Function with Supporting Invariant.}
We call a pair of a ranking function ${\color{rkBlue}f(\vec x)}$ and an inductive invariant ${\color{siGreen} I(\vec x)}$ of the lasso program $P$ a \emph{ranking function with supporting invariant} if the following holds.
\begin{itemize}
\item There exists a positive real number $\delta>0$ such that, if the inductive invariant holds then an execution of the loop decreases the value of the ranking function by at least $\delta$, i.e.,
\[
\exists\delta>0\forall \vec x\;\forall \vec x'\quad{\color{siGreen} I(\vec x)} \;\land\; \tau_\mathsf{loop}(\vec x,\vec x') \rightarrow {\color{rkBlue}f(\vec x)-f(\vec x')}\geq \delta.\tag{$\varphi_\mathsf{rkDecr}$}\label{rk-decr1}
\]
\item In states in which the inductive invariant holds and the loop can be executed, the value of the ranking function is non-negative, i.e., 
\[
\forall \vec x\;\forall \vec x'\quad{\color{siGreen} I(\vec x)} \;\land\; \tau_\mathsf{loop}(\vec x,\vec x') \rightarrow {\color{rkBlue}f(\vec x)\geq 0}.\tag{$\varphi_\mathsf{rkBound}$}\label{rk-bound1}
\]
\end{itemize}

For example, the lasso program depicted in \autoref{fig-bangalore} has the ranking function  ${\color{rkBlue}f(x,y)=x}$ with supporting invariant ${\color{siGreen} y\geq 1}$.

\paragraph{Linear lasso programs.}
Linear lasso programs. For the remainder of this paper we consider only linear lasso programs, linear ranking functions, and linear inductive invariants which we will define next.  The variables of the programs will range over the reals until we come to Section 9 where we turn to programs over integers.

\begin{definition}[Linear Lasso Program]
A \defo{linear lasso program} $$P=(\tau_\mathsf{stem},\tau_\mathsf{loop})$$ is a lasso program whose states are vectors over the reals, i.e. $\Sigma=\mathbb{R}^n$, and whose relations $\tau_\mathsf{stem}$ and $\tau_\mathsf{loop}$ are linear relations.
\end{definition}
We use the expression $A_\mathsf{stem}\cdot(^{\vec x}_{\vec x'})\leq \vec b_\mathsf{stem}$ to denote the relation $\tau_\mathsf{stem}$ of $P$.
We use the expression $A_\mathsf{loop}\cdot(^{\vec x}_{\vec x'})\leq \vec b_\mathsf{loop}$ to denote the relation $\tau_\mathsf{loop}$ of $P$.

\paragraph{Linear Ranking Function.}
If a ranking function $f:\mathbb{R}^n\rightarrow\mathbb{R}$ is an (affine) linear function, we call $f$ a \defo{linear ranking function}.
We use $r_1,\dots,r_n$ as coefficients of a linear ranking function, $\vec r$ as their vector,
$$f: \mathbb{R}^n \rightarrow \mathbb{R} \qquad {\color{rkBlue}f(\vec x)}={\color{rkBlue}\tr{\vec r} \cdot \vec x + r_0}.$$

\paragraph{Linear Invariant.}
If an inductive invariant $I(\vec x)$ is a linear predicate, we call $I$ a \defo{linear inductive invariant}. 
We use $s_1,\dots,s_n$ as coefficients of the term that defines the linear predicate, $\vec s$ as their vector,
$$ {\color{siGreen}I(\vec x)}\;\equiv\; {\color{siGreen}\tr{\vec s} \cdot \vec x + s_0\geq 0}.$$

\section[The Or-to-Plus Method]{The Or-to-Plus Method}

Our constraint-based method for the synthesis of linear ranking functions for linear lasso programs consists of three main steps:
\begin{description}
 \item[Step 1.] Set up four (universally quantified) constraints whose free variables are the coefficients of a linear ranking function with linear supporting invariant.
 \item[Step 2.] Apply Farkas' Lemma to the four constraints to obtain equivalent constraints without universal quantification.
 \item[Step 3.] Obtain solutions for the free variables by linear constraint solving.
\end{description}

The particularity of  our four constraints in Step~1
is that the application of Farkas' Lemma in Step~2 yields constraints that are linear.

Instead of presenting our constraints immediately, we derive them in three successive transformations of constraints.
We start with the four constraints \eqref{inv-stem1}, \eqref{inv-loop1}, \eqref{rk-decr1}, and \eqref{rk-bound1}.
Below, we have rephrased the four constraints for the setting where the ranking function is linear and the supporting invariant is linear.
We marked them \eqref{inv-stem2}, \eqref{inv-loop2}, \eqref{rk-decr2}, and \eqref{rk-bound2} in reference to Bradley, Manna and Sipma [5] who were the first to use them in the corresponding step of their method.

\subsubsection{The Bradley--Manna--Sipma constraints}\noindent

\noindent{\scriptsize for the special case of lasso programs and one linear supporting invariant}\footnote{
In~\cite{cav/BradleyMS05} the authors use more general general constraints that can be used to synthesize lexicographic linear ranking functions together with a conjunction of linear supporting invariants for programs that can also contains disjunctions.}
\begin{align}
\forall \vec x\;\forall \vec x'\qquad\qquad\qquad\qquad \tau_\mathsf{stem}(\vec x, \vec x') & \rightarrow {\color{siGreen}\tr{\vec\inv}\cdot\vec x' + \inv_0 \geq 0}\tag{$\varphi^\mathsf{\BMS}_1$}\label{inv-stem2}\\
\forall \vec x\;\forall \vec x'\quad {\color{siGreen}\tr{\vec\inv}\cdot\vec x + \inv_0 \geq 0} \;\land\; \tau_\mathsf{loop}(\vec x, \vec x') & \rightarrow {\color{siGreen}\tr{\vec\inv}\cdot\vec x' + \inv_0 \geq 0}\tag{$\varphi^\mathsf{\BMS}_2$}\label{inv-loop2}\\
\exists\delta>0\;\forall \vec x\;\forall \vec x'\quad {\color{siGreen}\tr{\vec\inv}\cdot\vec x + \inv_0 \geq 0} \;\land\; \tau_\mathsf{loop}(\vec x, \vec x') & \rightarrow {\color{rkBlue}\tr{\vec\rank}\cdot\vec x}-{\color{rkBlue}\tr{\vec\rank}\cdot\vec x'}\geq \delta\tag{$\varphi^\mathsf{\BMS}_3$}\label{rk-decr2} \\
\forall \vec x\;\forall \vec x'\quad {\color{siGreen}\tr{\vec\inv}\cdot\vec x + \inv_0 \geq 0} \;\land\; \tau_\mathsf{loop}(\vec x, \vec x') & \rightarrow {\color{rkBlue}\tr{\vec\rank}\cdot\vec x+\rank_0\geq 0} \tag{$\varphi^\mathsf{\BMS}_4$}\label{rk-bound2}
\end{align}
The free variables of
$\varphi^\mathsf{\BMS}_1\land\varphi^\mathsf{\BMS}_2\land\varphi^\mathsf{\BMS}_3\land\varphi^\mathsf{\BMS}_4$ are ${\color{rkBlue}\vec\rank}$, ${\color{rkBlue}\rank_0}$, ${\color{siGreen}\vec\inv}$, and ${\color{siGreen}\inv_0}$.

\subsubsection{Transformation 1: Move supporting invariant to right-hand side.}
We bring the conjunct
${\color{siGreen}\tr{\vec\inv}\cdot\vec x + \inv_0 \geq 0}$ 
in three of the four constraints \eqref{inv-stem2}, \eqref{inv-loop2}, \eqref{rk-decr2}, and \eqref{rk-bound2}
to the right-hand side of the implication, according to the following scheme.
$$\phi_1\land\phi_2\rightarrow\psi \;\;\equiv\;\; \phi_2\rightarrow\psi\lor\neg\phi_1$$
We obtain the following constraints.
\begin{align}
\forall \vec x\;\forall \vec x'\quad\tau_\mathsf{stem}(\vec x, \vec x') & \rightarrow {\color{siGreen}\tr{\vec\inv}\cdot\vec x' + \inv_0 \geq 0}\tag{$\psi_1$}\label{inv-stem3}\\
\forall \vec x\;\forall \vec x'\quad \tau_\mathsf{loop}(\vec x, \vec x') & \rightarrow  {\color{siGreen}\tr{\vec\inv}\cdot\vec x' + \inv_0 \geq 0}  \;\lor\; {\color{siGreen}-\tr{\vec\inv}\cdot\vec x - \inv_0 > 0}\tag{$\psi_2$}\label{inv-loop3} \\
\exists\delta>0\;\forall \vec x\;\forall \vec x'\quad \tau_\mathsf{loop}(\vec x, \vec x') & \rightarrow  {\color{rkBlue}\tr{\vec\rank}\cdot\vec x}-{\color{rkBlue}\tr{\vec\rank}\cdot\vec x'}\geq \delta \;\lor\; {\color{siGreen}-\tr{\vec\inv}\cdot\vec x - \inv_0 > 0} \tag{$\psi_3$}\label{rk-decr3}\\
\forall \vec x\;\forall \vec x'\quad \tau_\mathsf{loop}(\vec x, \vec x') & \rightarrow  {\color{rkBlue}\tr{\vec\rank}\cdot\vec x+\rank_0\geq 0} \;\lor\; {\color{siGreen}-\tr{\vec\inv}\cdot\vec x - \inv_0 > 0} \tag{$\psi_4$}\label{rk-bound3}
\end{align}

\subsubsection{Transformation 2: Drop supporting invariant in fourth constraint.}
We strengthen the fourth constraint \eqref{rk-bound3} by removing the disjunct $ {\color{siGreen}-\tr{\vec\inv}\cdot\vec x - \inv_0 > 0}$.
A solution for the strengthened constraint defines a ranking function
whose value is bounded from below for all states (and not just those that
satisfy the supporting invariant).

\subsubsection[Step 3: Replace disjunction by sum.]{Transformation 3: Replace disjunction by sum.}
We replace the disjunction on the right-hand side of the implication in constraints \eqref{inv-loop3} and \eqref{rk-decr3} by a single inequality, according to the scheme below. (It is the disjunction which prevents us from applying Farkas' Lemma to the constraints \eqref{inv-loop3} and \eqref{rk-decr3}.)
\begin{align*}
m\geq 0 \lor n> 0 \qquad \rightsquigarrow \qquad m+n\geq 0
\end{align*}
In the second constraint~\eqref{inv-loop3}, we replace the disjunction 
$${\color{siGreen}-\tr{\vec\inv}\cdot\vec x - \inv_0 > 0} \;\lor\; {\color{siGreen}\tr{\vec\inv}\cdot\vec x' + \inv_0 \geq 0}$$
by the inequality 
$${\color{siGreen}\tr{\vec\inv}\cdot\vec x' + \inv_0}{\color{siGreen}\;-\;\tr{\vec\inv}\cdot\vec x - \inv_0}\geq 0.$$
In the third constraint~\eqref{rk-decr3}, we replace the disjunction 
$${\color{siGreen}-\tr{\vec\inv}\cdot\vec x - \inv_0 > 0} \;\lor\; {\color{rkBlue}\tr{\vec\rank}\cdot\vec x}-{\color{rkBlue}\tr{\vec\rank}\cdot\vec x'}\geq\delta$$
by the inequality
$${\color{rkBlue}\tr{\vec\rank}\cdot\vec x}-{\color{rkBlue}\tr{\vec\rank}\cdot\vec x'} {\color{siGreen}-\tr{\vec\inv}\cdot\vec x - \inv_0}\geq\delta.$$
We obtain the following four constraints.
\subsubsection{The Or-to-Plus constraints}\noindent
\begin{align}
\forall \vec x\;\forall \vec x'\quad \tau_\mathsf{stem}(\vec x, \vec x') & \rightarrow {\color{siGreen}\tr{\vec\inv}\cdot\vec x' + \inv_0 \geq 0}\tag{$\varphi_1$}\label{inv-stem4}\\
\forall \vec x\;\forall \vec x'\quad \tau_\mathsf{loop}(\vec x, \vec x') & \rightarrow {\color{siGreen}\tr{\vec\inv}\cdot\vec x' + \inv_0}{\color{siGreen}\;-\;\tr{\vec\inv}\cdot\vec x - \inv_0}  \geq 0\tag{$\varphi_2$}\label{inv-loop4}\\
\exists\delta>0\;\forall \vec x\;\forall \vec x'\quad \tau_\mathsf{loop}(\vec x, \vec x') & \rightarrow {\color{rkBlue}\tr{\vec\rank}\cdot\vec x}-{\color{rkBlue}\tr{\vec\rank}\cdot\vec x'} {\color{siGreen}-\tr{\vec\inv}\cdot\vec x - \inv_0} \geq \delta\tag{$\varphi_3$}\label{rk-decr4}\\
\forall \vec x\;\forall \vec x'\quad \tau_\mathsf{loop}(\vec x, \vec x') & \rightarrow {\color{rkBlue}\tr{\vec\rank}\cdot\vec x+\rank_0\geq 0}\tag{$\varphi_4$}\label{rk-bound4}
\end{align}
The free variables of the conjunction $\varphi_1\land\varphi_2\land\varphi_3\land\varphi_4$ are ${\color{rkBlue}\vec\rank}$, ${\color{rkBlue}\rank_0}$, ${\color{siGreen}\vec\inv}$, and ${\color{siGreen}\inv_0}$.
\bigskip

Since we consider linear lasso programs, the relations $\tau_\mathsf{stem}$ and $\tau_\mathsf{loop}$ are given as conjunctions of linear constraints.
\begin{align*}
\tau_\mathsf{stem}(\vec x, \vec x') & \quad\equiv\quad A_\mathsf{stem}\cdot(^{\vec x}_{\vec x'})\leq \vec b_\mathsf{stem}\\
\tau_\mathsf{loop}(\vec x, \vec x') & \quad\equiv\quad A_\mathsf{loop}\cdot(^{\vec x}_{\vec x'})\leq \vec b_\mathsf{loop}
\end{align*}

\bigskip

We have now finished the description for the three transformation steps that lead us to the the or-to-plus constraints. We are now ready to introduce our method.

\medskip

\begin{center}
\fbox{
 \begin{minipage}{11cm}
 \textbf{The Or-to-Plus Method}

  \begin{description}
  \item[Input:] linear lasso program $P$.
  \item[Output:] coefficients ${\color{rkBlue}\vec\rank}$, ${\color{rkBlue}\rank_0}$, ${\color{siGreen}\vec \inv}$, and ${\color{siGreen}\inv_0}$ of a linear ranking function with linear supporting invariant
 \end{description}

 \begin{enumerate}
  \item Set up the or-to-Plus constraints \ref{inv-stem4},  \ref{inv-loop4}, \ref{rk-decr4}, and \ref{rk-bound4} for $P$.
  \item Apply Farkas' Lemma to each constraint.
  \item Obtain ${\color{rkBlue}\vec\rank}$, ${\color{rkBlue}\rank_0}$, ${\color{siGreen}\vec \inv}$, and ${\color{siGreen}\inv_0}$,  by linear constraint solving.
 \end{enumerate}
 \end{minipage}
}
\end{center}

\medskip

After setting up the four or-to-plus constraints $\varphi_1$, $\varphi_2$, $\varphi_3$, $\varphi_4$ in Step~1, we apply Farkas' Lemma to each of the four constraints in Step~2.
We obtain four linear constraints.
E.g., by applying Farkas' Lemma to the constraint~\eqref{rk-decr4} we obtain the following linear constraint.
\begin{align*}
\exists\delta>0\quad\exists \vec\lambda \quad \vec\lambda\geq 0 \quad\land\quad \tr{\vec\lambda}\cdot A_\mathsf{loop} = {\color{rkBlue}\tr{\abovebelow{{\color{siGreen}\vec s} -\vec\rank}{\phantom{-}\vec\rank}}} \quad\land\quad \tr{\vec\lambda} \cdot \vec b_\mathsf{loop}\leq -\delta - {\color{siGreen} s_0}
\end{align*}

We apply linear constraint solving in Step~3. We obtain a satisfying assignment for the free variables in the resulting constraints.
The values obtained for ${\color{rkBlue}\vec\rank}$, ${\color{rkBlue}\rank_0}$, ${\color{siGreen}\vec \inv}$ and ${\color{siGreen}\inv_0}$  are the coefficients of a linear ranking function ${\color{rkBlue}f(\vec x})$ with linear supporting invariant ${\color{siGreen}I(\vec x})$.

The or-to-plus method inherits its soundness from method of Bradley--Manna--Sipma. 
Step~1 is an equivalence
transformation on the Bradley--Manna--Sipma constraints, Step~2 and
Step~3 strengthen the constraints, and the application of Farkas' Lemma
is an equivalence transformation.   Thus, a satisfying assignment of the
or-to-plus constraints obtained after the application of Farkas' Lemma is also a satisfying assignment of the Bradley--Manna--Sipma constraints.

\section{Completeness of the Or-to-Plus Method}

In the tradition of constraint-based synthesis for verification, we
will formulate completeness according to the following scheme: the method
\texttt{X} applied to a program $P$ in the class \texttt{Y} computes (the
coefficients of) a correctness argument of the form \texttt{Z}
whenever one exists (i.e., whenever a correctness argument of the form
\texttt{Z} exists for the program $P$).  Here, \texttt{X} is the
or-to-plus method, \texttt{Y} is the class of lasso programs, and
\texttt{Z} is a termination argument consisting of a linear ranking
function and an invariant of a form that we we define next.

\begin{definition}[Non-decreasing linear inductive invariant]
We call a linear inductive invariant ${\color{siGreen}\tr{\vec s} \cdot \vec x + s_0 \geq 0}$ of the lasso program P \emph{non-decreasing} if the loop implies that the value of the term ${\color{siGreen}\tr{\vec s} \cdot \vec x + s_0}$ does not decrease when executing the loop, i.e.,
$$ \tau_\mathsf{loop} \rightarrow{\color{siGreen} \tr{\vec s} \cdot \vec x'}\geq {\color{siGreen} \tr{\vec s} \cdot \vec x}.$$
\end{definition}

In \autoref{sec:examples} we give examples which may help to convey some
intuition about the meaning of `non-decreasing', examples of those
terminating programs that do have a linear ranking function with a
non-decreasing linear supporting invariant, and examples of those that
don't.  

\begin{figure}[t]
\begin{center}
\begin{minipage}{4cm}
\begin{verbatim}
x := y + 42;
while( x >= 0 ) {
    y := 2*y - x;
    x := (y + x) / 2;
}
\end{verbatim}

\end{minipage}
\begin{minipage}{6cm}
\vspace{-8mm}
\[
\begin{array}{ll}
 \tau_\mathsf{stem}: & x'=y+42\;\land\;y'=y\\[2mm]
 \tau_\mathsf{loop}: & x\geq 0 \;\land\; x'=y \;\land\; y'=2y-x
\end{array}
\]
\end{minipage}
\end{center}
\caption{Linear lasso program $P_\mathsf{diff42}=(\tau_\mathsf{stem},\tau_\mathsf{loop})$ that has the linear ranking function ${\color{rkBlue}f(x,y)=x}$ with linear supporting invariant ${\color{siGreen}x-y\geq 42}$. 
}
\label{fig-nonObvious}
\end{figure}

\begin{theorem}[Completeness]\label{thm-completeness}
The or-to-plus method applied to the linear lasso program $P$ succeeds and computes the coefficients of a linear ranking function with non-decreasing linear supporting invariant whenever one exists.
\end{theorem}

To prove this theorem we use the following lemma.

\begin{lemma}\label{lem-halfspaces}
Given are
\begin{enumerate}
\item[(1)] satisfiable linear inequalities $A\cdot \vec x \leq \vec b$,
\item[(2)] an inequality $\tr{\vec g} \cdot \vec x + g_0\geq 0$, and
\item[(3)] a strict inequality $\tr{\vec h} \cdot \vec x + h_0> 0$.
\end{enumerate}
If $A \cdot  \vec x \leq \vec b$ does not imply the strict inequality
(3), but the disjunction of (2) and (3), i.e.\
$$\forall\vec x\quad A\cdot  \vec x \leq \vec b \;\rightarrow\; \tr{\vec g}\cdot  \vec x + g_0\geq 0 \;\lor\; \tr{\vec h}\cdot  \vec x + h_0> 0,$$
then there exists a constant $\mu \geq 0$ such that
$$\forall\vec x\quad A\cdot  \vec x \leq \vec b \;\rightarrow\;
(\tr{\vec g}\cdot \vec x + g_0) + \mu \cdot (\tr{\vec h}\cdot \vec x + h_0)\geq 0.$$
\end{lemma}

\begin{figure}
\begin{minipage}{7cm}
\begin{tikzpicture}[scale=0.8]


\draw[color=orange] (-4,-3) -- (4,3);
\fill[color=orange,opacity=0.2] (-4,-3) -- (4,3) -- (4,3) -- (-4,3);
\node at (-2.6, -0.5) {$\tr{\vec g}\cdot \vec x + g_0 \geq 0$};

\fill[color=violet,opacity=0.2] (-4,1) -- (4,-1) -- (4,3) -- (-4,3);
\node at (2.5, 0.1) {$\tr{\vec h}\cdot \vec x + h_0 > 0$};

\filldraw[black] (0,0) circle (1mm);
\node at (0,-0.4) {$Z$};

\draw[dashed] (-2, 0.2) -- (2, 0.8) -- (1, 2) -- (-1, 2.8) -- (-3, 2) -- (-2, 0.2);
\node at (-1, 1.6) {$A\cdot \vec x \leq b$};
\end{tikzpicture}
\end{minipage}
\begin{minipage}{5cm}
Let $H = \{ \vec x \,|\, \tr{\vec g}\cdot \vec x + g_0 \geq 0\}$,
and $H' = \{ \vec x \,|\, \tr{\vec h}\cdot \vec x + h_0 > 0 \}$ be
half-spaces defined by linear inequalities.
A half-space $H_\mu = \{ \vec x \,|\, (\tr{\vec g}\cdot \vec x + g_0)
+ \mu \cdot (\tr{\vec h}\cdot \vec x + h_0) \geq 0 \}$ defined by a weighted
sum is a rotation of $H$ around the intersection $Z$ of the boundary
of $H$ and the boundary of $H'$.

If a polyhedron $X$ is contained in the union $H \cup H'$, then there
is a half-space $H_\mu$ defined by a weighted sum that contains $X$.
\end{minipage}

\label{fig:geometric}
\caption{A geometrical interpretation of \autoref{lem-halfspaces}.}
\end{figure}

\begin{proof}[of \autoref{lem-halfspaces}]
\[
\forall\vec x \quad A\cdot\vec x \leq \vec b \;\rightarrow\; (\tr{\vec g}
\cdot\vec x + g_0 \geq 0 \;\lor\; \tr{\vec h}\cdot\vec x + h_0 > 0)
\]
is equivalent to
\[
\forall\vec x \quad (A\cdot\vec x \leq \vec b
\;\land\; \tr{\vec h}\cdot\vec x + h_0 \leq 0) \;\rightarrow\; \tr{\vec g}
\cdot\vec x + g_0 \geq 0.
\]
By assumption, {\it (1)} does not imply {\it (3)}, so $A\cdot\vec x
\leq \vec b \;\land\; \tr{\vec h}\cdot\vec x + h_0 \leq 0$ is
satisfiable, and by Farkas' Lemma this formula is equivalent to
\[
\exists \mu \geq 0 \;\exists \vec\lambda \geq 0
\quad \mu \cdot \tr{\vec h} + \tr{\vec\lambda}\cdot A = -\tr{\vec g}
\;\land\; \tr{\vec\lambda}\cdot\vec b + \mu \cdot (-h_0) \leq g_0,
\]
and thus
\[
\exists \mu \geq 0 \;\exists \vec\lambda \geq 0
\quad \tr{\vec\lambda}\cdot A = -(\mu \cdot \tr{\vec h} + \tr{\vec g})
\;\land\; \tr{\vec\lambda}\cdot\vec b \leq \mu \cdot h_0 + g_0.
\]
Because $A\cdot\vec x \leq \vec b$ is satisfiable by assumption,
Farkas' Lemma can be applied again to yield
$$\exists \mu \geq 0 \;\forall\vec x \qquad A\cdot\vec x \leq \vec b
\;\rightarrow\; -(\mu \cdot \tr{\vec h} + \tr{\vec g})\vec x \leq \mu
\cdot h_0 + g_0. \quad\qed$$
\end{proof}

\begin{proof}[of \autoref{thm-completeness}]
Let $f(\vec x) = \tr{\vec \ranki}\cdot \vec x + \ranki_0$ be a ranking
 function with non-decreasing supporting invariant $I(\vec x) \equiv \tr{\vec\invi}\cdot \vec x + \invi_0 \geq 0$ for the lasso program $P$.
Since executions of our lasso programs comprise at least three states,
there can be no supporting invariant that contradicts the loop, i.e.\
\begin{align}
A_\mathsf{loop}\cdot(^{\vec x}_{\vec x'})\leq \vec b_\mathsf{loop}
\;\rightarrow\; -\tr{\vec{\invi}}\cdot \vec x - \invi_0 > 0
\label{eq-completeness-executable}
\end{align}
is not valid.
From \eqref{rk-bound1} it follows that
$$\tr{\vec{\invi}}\cdot \vec x + \invi_0\geq 0 \;\land\;
A_\mathsf{loop}\cdot(^{\vec x}_{\vec x'})\leq \vec b_\mathsf{loop}
\;\rightarrow\;  \tr{\vec{\ranki}}\cdot \vec x+\ranki_0 \geq 0,$$
and hence the implication
$$A_\mathsf{loop}\cdot(^{\vec x}_{\vec x'})\leq \vec b_\mathsf{loop}
\;\rightarrow\;  \tr{\vec{\ranki}}\cdot \vec x+\ranki_0\geq 0 \;\lor\;
-\tr{\vec{\invi}}\cdot \vec x - \invi_0 > 0$$
is valid.
By \eqref{eq-completeness-executable} and \autoref{lem-halfspaces}
there is a $\mu_1 \geq 0$ such that
$$A_\mathsf{loop}\cdot(^{\vec x}_{\vec x'})\leq \vec b_\mathsf{loop}
\;\rightarrow\; (\tr{\vec{\ranki}}\cdot \vec x+\ranki_0) \;+\;
\mu_1\cdot(-\tr{\vec{\invi}}\cdot \vec x - \invi_0) \geq 0$$ is valid.
If we assign $\vec \rank\mapsto \vec \ranki - \mu_1 \cdot \invi,
\rank_0\mapsto \ranki_0 - \mu_1 \cdot \invi_0$, then \eqref{rk-bound4}
is satisfied.

Because $I(\vec x) \equiv \vec\invi \cdot \vec x + \invi_0 \geq 0$ is
a non-decreasing invariant,
$$A_\mathsf{loop}\cdot(^{\vec x}_{\vec x'})\leq \vec b_\mathsf{loop}
\;\rightarrow\; \tr{\vec{\invi}}\cdot \vec x'-\tr{\vec{\invi}}\cdot \vec x
\geq 0,$$ and hence, since $\mu_1 \geq 0$,
\begin{align}
A_\mathsf{loop}\cdot(^{\vec x}_{\vec x'})\leq \vec b_\mathsf{loop}
\;\rightarrow\; -\mu_1 \cdot \tr{\vec{\invi}} \cdot (\vec x - \vec x') \geq 0.
\label{eq-completeness-si}
\end{align}

From \eqref{rk-decr1} we know that
$$\tr{\vec{\invi}}\cdot \vec x + \invi_0 \geq 0 \;\land\;
A_\mathsf{loop}\cdot(^{\vec x}_{\vec x'})\leq \vec b_\mathsf{loop}
\;\rightarrow\; \tr{\vec{\ranki}}\cdot \vec x-\tr{\vec{\ranki}}\cdot \vec x'
\geq \delta,$$ and hence equivalently
$$A_\mathsf{loop}\cdot(^{\vec x}_{\vec x'})\leq \vec b_\mathsf{loop}
\;\rightarrow\; \tr{\vec{\ranki}}\cdot \vec x-\tr{\vec{\ranki}}\cdot \vec
x'\geq\delta \;\lor\; -\tr{\vec{\invi}}\cdot \vec x -\invi_0 > 0.$$
With \eqref{eq-completeness-si} we obtain validity of the following
formula.
$$A_\mathsf{loop}\cdot(^{\vec x}_{\vec x'})\leq \vec b_\mathsf{loop}
\;\rightarrow\; (\tr{\vec{\ranki}}-\mu_1\cdot\tr{\vec{\invi}})\cdot (\vec x-\vec
x') \geq\delta \;\lor\; -\tr{\vec{\invi}}\cdot \vec x -\invi_0 > 0$$
By \eqref{eq-completeness-executable} and \autoref{lem-halfspaces}
there exists a $\mu_2 \geq 0$ such that
$$A_\mathsf{loop}\cdot(^{\vec x}_{\vec x'})\leq \vec b_\mathsf{loop}
\;\rightarrow\; (\tr{\vec{\ranki}}-\mu_1\cdot\tr{\vec{\invi}})\cdot (\vec x-\vec
x') + \mu_2\cdot (-\tr{\vec{\invi}}\cdot \vec x -\invi_0) > \delta.$$
We pick the assignment $\vec \rank \mapsto \vec \ranki -
\mu_1 \cdot \invi, \rank_0 \mapsto \ranki_0 - \mu_1 \cdot \invi_0,
\vec \inv \mapsto \mu_2 \cdot \vec\invi, \inv_0 \mapsto \mu_2 \cdot
\invi_0$, which hence satisfies \eqref{rk-decr4}.  We already argued
that it satisfies \eqref{rk-bound4}, and from $\mu_2 \geq 0$ and
the fact that $I(\vec x)$ is a non-decreasing inductive invariant
it follows that the assignment also satisfies \eqref{inv-stem4} and
\eqref{inv-loop4}.
Hence, the ranking function ${\color{rkBlue}\tr{(\vec \ranki -\mu_1 \cdot \invi)}\cdot \vec x + \ranki_0 - \mu_1 \cdot \invi_0}$ with supporting invariant ${\color{siGreen}\tr{(\mu_2 \cdot \vec\invi)}\cdot \vec x + \mu_2 \cdot\invi_0\geq 0}$ can be found by the or-to-plus method.
\qed
\end{proof}

\section{Examples} \label{sec:examples}

\begin{wrapfigure}{r}{0.4\textwidth}
  \vspace{-30pt}
  \begin{center}
  \begin{minipage}{0.3\textwidth}
\begin{verbatim}
y := 23;
while( x >= y ) {
    x := x - 1;
}
\end{verbatim}
\end{minipage}
  \end{center}
  \vspace{-10pt}
  \caption{Lasso program $P_\mathsf{bound}$}
  \vspace{-10pt}
  \label{fig-other}
\end{wrapfigure}

Our three transformations strengthened the Bradley--Manna--Sipma
constraints, hence the solution space of the or-to-plus constraints is
smaller than the solution space of the Bradley--Manna--Sipma
constraints. This can be seen e.g., in the example depicted in
\autoref{fig-other}. The program $P_\mathsf{bound}$ has the linear ranking function
${\color{rkBlue}f(x,y)=x}$ with linear supporting invariant
${\color{siGreen}y\geq 23}$, but the coefficients of this ranking
function and supporting invariant are no solution of the or-to-plus constraints; the constraint \ref{rk-bound4} is violated.
Does this mean that our method will not succeed?
No, it does not.  By \autoref{thm-completeness}, in fact, we do know that the method will succeed. 
I.e., since we know of some linear ranking function with non-decreasing supporting invariant  (in this case, ${\color{rkBlue}f(x,y)=x}$  and ${\color{siGreen}y\geq 23}$), even if it is not a solution, we know that there exists one which is a solution (here, for example, ${\color{rkBlue}f(x,y)=x-y}$ with the (trivial) supporting invariant ${\color{siGreen}0\geq 0}$).

\begin{wrapfigure}{l}{0.35\textwidth}
  \vspace{-10pt}
  \begin{center}
\begin{center}
\begin{minipage}{0.35\textwidth}
\begin{verbatim}
y := 2;
while( x >= 0 ) {
    x := x - y;
    y := (y + 1) / 2;
}
\end{verbatim}
\end{minipage}
\end{center}
  \end{center}
  \vspace{-10pt}
  \caption{Lasso program $P_\mathsf{zeno}$}
  \vspace{-10pt}
  \label{fig-zeno}
\end{wrapfigure}

The prerequisite of \autoref{thm-completeness} is the existence of a non-decreasing supporting invariant. 
There are linear lasso programs that have a linear ranking function with linear supporting invariant, but do not have a linear ranking function with a non-decreasing linear supporting invariant. E.g., for the lasso programs depicted in \autoref{fig-zeno} and \autoref{fig-havoc} our or-to-plus method is not able to synthesize a ranking function for these programs.

The linear lasso program $P_\mathsf{zeno}$ depicted in \autoref{fig-zeno}
has the linear ranking function ${\color{rkBlue}f(x,y)=x}$ with the linear supporting invariant ${\color{siGreen}y\geq 1}$. However this inductive invariant is not non-decreasing; while executing the loop the value of the variable ${\color{siGreen}y}$ converges to $1$ in the following sequence. $2,\; 1+\frac{1}{2},\; 1+\frac{1}{4},\; 1+\frac{1}{8},\dots$.

\begin{wrapfigure}{r}{0.35\textwidth}
  \vspace{-20pt}
  \begin{center}
\begin{center}
\begin{minipage}{0.3\textwidth}
\begin{verbatim}
assume y >= 1;
while( x >=0 ) {
    x := x - y;
    havoc y;
    assume (y >= 1);
}
\end{verbatim}
\end{minipage}
\end{center}
  \end{center}
  \vspace{-10pt}
  \caption{Lasso program $P_\mathsf{wild}$}
  \vspace{-10pt}
  \label{fig-havoc}
\end{wrapfigure}

The statement \ \texttt{havoc y;}\ in the lasso program $P_\mathsf{wild}$  is a nondeterministic assignment to the variable $y$. The relations $\tau_\mathsf{stem}$ and $\tau_\mathsf{loop}$ of this lasso program are given by the constraints $y'\geq 1$ and $x\geq 0 \;\land\; x'=x-y \;\land\; y'\geq 1$.
$P_\mathsf{wild}$ has the ranking function ${\color{rkBlue}f(x,y)=x}$ with the supporting invariant ${\color{siGreen}y\geq 1}$, however this inductive invariant is not non-decreasing in each execution of the loop the variable ${\color{siGreen}y}$ can get any value greater than or equal to one.

The next example shows that nondeterministic updates are no general obstacle for our or-to-plus method. In the linear lasso program $P_\mathsf{array}$ the loop iterates over an array of positive integers.  The index accessed in the next
\begin{wrapfigure}{l}{0.43\textwidth}
  \vspace{-20pt}
  \begin{center}
\begin{center}
\begin{minipage}{0.43\textwidth}
\begin{verbatim}
offset := 1;
i := 0;
while(i<=a.length) {
    assume a[i]>=0;
    i := i + offset + a[i];
}
\end{verbatim}
\end{minipage}
\end{center}
  \end{center}
  \vspace{-10pt}
  \caption{Lasso program $P_\mathsf{array}$}
  \vspace{-10pt}
  \label{fig-array}
\end{wrapfigure}
 iteration is the sum of the current index, the current entry of the array, and an offset. The relations $\tau_\mathsf{stem}$ and $\tau_\mathsf{loop}$ of this lasso program are given by the constraints $\mathit{offset}'= 1 \;\land\; i'=0$ and $i\leq \mathit{a.length} \;\land\; \mathit{curVal}'\geq 0 \;\land\; i'=i+\mathit{offset}+\mathit{curVal}'$. The variable $\mathit{curVal}$ which represents the current entry of the array \texttt{a[i]} can get any value greater than or equal to one in each loop iteration. The or-to-plus method finds the linear ranking function ${\color{rkBlue}f(i,\mathit{offset})=i-\mathit{a.length}}$ with the linear supporting invariant ${\color{siGreen}\mathit{offset}\geq 1}$.

\section{Lasso Programs over the Integers}

In the preceding sections we considered lasso programs over the
reals.  In this section we discuss the applicability of the or-to-plus method to linear lasso programs over the integers, i.e., programs where the set of states $\Sigma$ is a subset of $\mathbb{Z}^n$.  We still use real-valued ranking functions. We obtain the constraints for coefficients of a linear ranking function with linear supporting invariant by restricting the range of the universal quantification in the constraints \ref{inv-stem4},  \ref{inv-loop4}, \ref{rk-decr4}, and \ref{rk-bound4} to the integers. E.g., the constraint \ref{rk-decr4} for linear lasso programs over the integers is
\begin{align*}
\exists\delta>0\;\forall \vec x\in\mathbb{Z}^n\;\forall \vec x'\in\mathbb{Z}^n\quad  \tau_\mathsf{loop}(\vec x, \vec x') & \rightarrow {\color{rkBlue}\tr{\vec\rank}\cdot\vec x}-{\color{rkBlue}\tr{\vec\rank}\cdot\vec x'} {\color{siGreen}-\tr{\vec\inv}\cdot\vec x - \inv_0} \geq \delta
\end{align*}
where the domain of the coefficients ${\color{rkBlue}\vec\rank}$, ${\color{rkBlue}\rank_0}$, ${\color{siGreen}\vec \inv}$, and ${\color{siGreen}\inv_0}$ and the quantified variable $\delta$ are the reals.
Now, Farkas' lemma is not an equivalence transformation, its application results in weaker formulas.
This means the or-to-plus method is still 
sound, but we loose the completeness result of \autoref{thm-completeness}. An example for this is 
\begin{wrapfigure}{l}{0.40\textwidth}
  \vspace{-20pt}
  \begin{center}
\begin{center}
\begin{minipage}{0.40\textwidth}
\begin{verbatim}
assume 2*y >= 1;
while( x >= 0 ) {
    x := x - 2*y + 1;
}
\end{verbatim}
\end{minipage}
\end{center}
  \end{center}
  \vspace{-10pt}
  \caption{Lasso program $P_\mathsf{nonIntegral1}$}
  \vspace{-20pt}
  \label{fig-onlyInt1}
\end{wrapfigure}
the program $P_\mathsf{nonIntegral}$, depicted in \autoref{fig-onlyInt1} that has the following transition relations.
\[
\begin{array}{ll}
 \tau_\mathsf{stem}: & 2y'\geq 1 \;\land\; x'=x\\[2mm]
 \tau_\mathsf{loop}: & x\geq 0 \;\land\; x'=x - 2y + 1 \;\land\; y'=y;
\end{array}
\]
Over integer variables, $P_\mathsf{nonIntegral1}$ has the linear ranking function ${\color{rkBlue}f(x,y)=x}$ with the linear supporting invariant ${\color{siGreen}y\geq 1}$. 
Over real-valued variables, $P_\mathsf{nonIntegral1}$ does not terminate.
If we add the additional constraint $y'\geq 1$ to $\tau_\mathsf{stem}$, the programs' semantics over the integers is not changed, but we are able to synthesize a linear ranking function with a linear supporting invariant.
Adding this additional constraint gives the constraints a property that we formally define as follows.

\paragraph{Integral constraints.} A conjunction of linear constraints
$A\cdot \vec x\leq \vec b$ is called \defo{integral} if the set of
satisfying assignments over the reals $S := \{\vec r \in \mathbb{R}^n \mid A\cdot \vec r\leq \vec b\}$ coincides with the integer hull of $S$ (the convex hull of all integer vectors in $S$).

\medskip

For each conjunction of $m$ linear constraints there is an equivalent
conjunction of at most $2^m$ linear constraints that is
integral~\cite{Schrijver:1986:TLI:17634}. We add an
additional step to the or-to-plus method in which we make the constraints in the
stem transition $\tau_\mathsf{stem}$ and loop transition $\tau_\mathsf{loop}$ integral.

\begin{center}
\fbox{
 \begin{minipage}{11cm}
 \textbf{The Or-to-Plus Method (Int)}

  \begin{description}
  \item[Input:] linear lasso program $P$ with integer variables
  \item[Output:] coefficients ${\color{rkBlue}\vec\rank}$, ${\color{rkBlue}\rank_0}$, ${\color{siGreen}\vec \inv}$, and ${\color{siGreen}\inv_0}$ of linear ranking function with linear supporting invariant 
 \end{description}

 \begin{enumerate}
  \item Replace $\tau_\mathsf{stem}$ and $\tau_\mathsf{loop}$ by equivalent integral linear constraints.
  \item Set up constraints \ref{inv-stem4},  \ref{inv-loop4}, \ref{rk-decr4}, and \ref{rk-bound4} for $P$.
  \item Apply Farkas' Lemma to each constraint.
  \item Obtain ${\color{rkBlue}\vec\rank}$, ${\color{rkBlue}\rank_0}$, ${\color{siGreen}\vec \inv}$, and ${\color{siGreen}\inv_0}$,  by linear constraint solving.
 \end{enumerate}
 \end{minipage}
}
\end{center}

That we find more solutions after making the linear constraints $\tau_\mathsf{stem}$ and
$\tau_\mathsf{loop}$ integral is due to the following lemma which was stated in~\cite{journals/fmsd/CookKRW13}. We present our proof for the purpose of self-containment.

\begin{lemma}[Integral version of Farkas' Lemma]\label{lem-integral-farkas}
Given a conjunction of linear constraints $A\cdot \vec x\leq \vec b$
which is satisfiable and integral, and a linear constraint $\tr{\vec c}\cdot \vec x \leq \delta$, 
\begin{center}
$\forall \vec x\in\mathbb{Z}^n \;\;  (A\cdot \vec x\leq \vec b \rightarrow \tr{\vec c}\cdot \vec x \leq \delta)$ \ \ \ iff \ \ \  $\exists \vec\lambda\;\;(\vec\lambda\geq 0 \land \tr{\vec\lambda}\cdot A = \tr{\vec c} \land \tr{\vec\lambda}\cdot \vec b\leq \delta)$
\end{center}
\end{lemma}
\begin{proof}
We write this statement as a linear programming problem.
$$\textbf{(P)} \qquad \max \{ \tr{\vec c} \cdot \vec x \;|\; A\cdot\vec x \leq \vec b \}
\leq \delta$$
Because the constraints $A \cdot \vec x \leq \vec b$ are integral,
there is an integral vector $\vec x \in \mathbb{Z}^n$ such that
$\tr{\vec c} \cdot \vec x$ is the optimum solution to {\bf (P)}.
Thus the optimum over integers is $\leq \delta$ if and only if the
optimum of the reals is.  The statement now follows from the
real version of Farkas' Lemma.
\qed
\end{proof}

\begin{wrapfigure}{r}{0.45\textwidth}
  \vspace{-32pt}
\begin{flushright}
\begin{minipage}{0.42\textwidth}
\begin{verbatim}
assume 2*y >= z;
while( x >= 0 && z == 1 ) {
    x := x - 2*y + 1;
}
\end{verbatim}
\end{minipage}
\end{flushright}
  \vspace{-10pt}
  \caption{Lasso program $P_\mathsf{nonIntegral2}$}
  \vspace{-15pt}
  \label{fig-onlyInt2}
\end{wrapfigure}

However, even if $\tau_\mathsf{stem}$ and $\tau_\mathsf{loop}$ are integral, our
 method is not complete over the integers.
In the \hyperref[thm-completeness]{completeness proof} for the reals 
we applied Farkas' Lemma to conjunctions of a polyhedron 
$A \cdot \vec x \leq \vec b$ and an inequality $\tr{\vec h} \cdot x + h_0 \leq 0$.
This inequality contains free variables, namely the coefficients of the 
supporting invariant ${\color{siGreen}\tr{\vec s} \cdot \vec x + s_0 \geq 0}$.
Even if $\tau_\mathsf{stem}$ and $\tau_\mathsf{loop}$ are integral, this 
conjunction might not be integral and we cannot apply the integer 
version of Farkas' lemma in this case. 

A counterexample to completeness of our integer version of the or-to-plus method is the linear lasso program $P_\mathsf{nonIntegral2}$ depicted in \autoref{fig-onlyInt2}.

\section{Conclusion}

We have presented a constraint-based synthesis method for a  class
of programs that was not investigated before for the synthesis problem.
The class is restricted (though less restricted than the widely studied class of simple
while programs) but still requires the combined synthesis of not only
a ranking function but also an invariant.
We have formulated and proven a completeness theorem that gives us an
indication on the extent of power of a method that does without nonlinear constraint solving.

We implemented the or-to-plus method as plugin of the
\textsc{Ultimate}
 software analysis framework.
A version that allows one to `play around' with lasso programs is available via a web interface at the following URL.
\begin{center}
 \texttt{\url{http://ultimate.informatik.uni-freiburg.de/LassoRanker}}
\end{center}

As mentioned in the introduction, the class of lasso programs is
motivated by the fact that they are a natural way (and, it seems, the
only way) to represent an (infinite) counterexample path in a control flow graph.
It is a topic of future research to explore the different scenarios in
practical tools that use a module 
to  find a ranking
function and a supporting invariant for a lasso program
(e.g., in
\cite{DBLP:journals/jar/AlbertAGP11,cav/CookPR06,conf/pldi/GulwaniZ10,conf/popl/GuptaHMRX08,sas/HarrisLNR10,lics/PodelskiR04,hybrid/PodelskiW07})
and to compare the performance of our---theoretically
motivated---synthesis method in comparison with the
existing---heuristically motivated---approach used presently in the module.

\bibliographystyle{abbrv}
\bibliography{main}

\end{document}